\documentclass{fac}

\usepackage{graphicx}
\usepackage{amssymb}
\usepackage{amsfonts}
\usepackage{amsmath}
\usepackage{dsfont}
\usepackage{MnSymbol}
\usepackage{mathtools}
\usepackage{epsfig}
\usepackage{epstopdf}

\newtheorem{theorem}{Theorem}[section]

\newtheorem{proposition}[theorem]{Proposition}

\title[Draft of Multi-objective Non-cooperative Game Model for Cost-based Task Scheduling in Computational Grid]
      {Multi-objective Non-cooperative Game Model for Cost-based Task Scheduling in Computational Grid}

\author[Ziyan Gao, Yong Wang, Yifan Gao, Xingtian Ren]
    {Ziyan Gao, Yong Wang, Yifan Gao, Xingtian Ren\\
     College of Computer Science and Technology,\\
     Faculty of Information Technology,\\
     Beijing University of Technology, Beijing, China\\
     }

\correspond{Yong Wang, Pingleyuan 100, Chaoyang District, Beijing, China.
            e-mail: wangy@bjut.edu.cn}

\pubyear{2017}
\pagerange{\pageref{firstpage}--\pageref{lastpage}}

\begin{document}
\label{firstpage}

\makecorrespond

\maketitle

\begin{abstract}
Task scheduling is an important and complex problem in computational grid. A computational grid often covers a range of different kinds of nodes, which offers a complex environment. There is a need to develop algorithms that can capture this complexity as while as can be easily implemented and used to solve a wide range of load-balancing scenarios. In this paper, we propose a game theory based algorithm to the grid load balancing problem on the principle of minimizing the cost of the grid. The grid load-balancing problem is treated as a noncooperative game. The experiment results demonstrate that the game based algorithm has a desirable capability to reduce the cost of the grid.
\end{abstract}


\section{Introduction}
Task scheduling is the key problem in computational grid research, which commonly studies on task allocation. Task allocation aims at getting a fairness of load between each computational node, while minimizing the average cost of task run, or minimizing the average task execution time, which generally considered as to reach a best performance. Different with traditional distributed system, computational grid has its own characteristics, which cost or execution time of tasks running on computational nodes is only one side in the whole system. In fact, computational grid often coves a wide scope, cost on communication is a considerable factor in the system cost.

In this paper, we propose a game theoretic based solution to the grid load balancing problem. Different with other research, our research focuses on minimizing the cost of system when executing tasks, but not minimizing the execution time of tasks. In this game theory based solution to the grid load balancing problem, we make the cost as the object of game, and the slice strategy on each scheduler as the game strategy.

In general, job allocation algorithms in distributed systems can be classified as static or dynamic \cite{2}. In static algorithms, job allocation decisions are made at compile time and remain constant during runtime. For example, in \cite{3}, Kim and Kameda proposed a simplified load balancing algorithm, which targets at the minimizing the overall mean job response time via adjusting the each node's load in a distributed computer system that consists of heterogeneous hosts, based on the single-point algorithm originally presented by Tantawi and Towsley. Grosu and Leung\cite{4} formulated a static load balancing problem in single class job distributed systems from the aspect of cooperative game among computers. Also, there exists several studies on static load balancing in multi-class job systems \cite{5,6}. In contrast, dynamic job allocation algorithms attempt to use the runtime state information to make more informative job allocation decisions. In \cite{7}, Delavar introduced a new scheduling algorithm for optimal scheduling of heterogeneous tasks on heterogeneous sources, according to Genetic Algorithm which can reach to better makespan and more efficiency. In \cite{8}, Fujimoto proposed a new algorithm RR that uses the criterion called total processor cycle consummation, which is the total number of instructions the grid could compute until the completion time of the schedule, regardless how the speed of each processor varies over time, the consumed computing power can be limited within $(1+m(\ln(m-1)+1)/n)$( $m$ represents the number of the processor, n represents the number of independent coarse-grained tasks with the same length)times the optimal one.

For balanced task scheduling, \cite{9,10,11} proposed some models and task scheduling algorithms in distributed system with the market model and game theory. \cite{12,13} introduced a balanced grid task scheduling model based on non-cooperative game. QoS-based grid job allocation problem is modeled as a cooperative game and the structure of the Nash bargaining solution is given in \cite{2}. In \cite{14}, Wei and Vasilakos presented a game theoretic method to schedule dependent computational cloud computing services with time and cost constrained, in which the tasks are divided into subtasks. The above works generally take the scheduler or job manager as the participant of the game, take the total execution time of tasks as the game optimization goals and give the proof of the existence of the Nash equilibrium solution and the solving Nash equilibrium solution algorithm, or model the task scheduling problem as a cooperative game and give the structure of the cooperative game solution.

\section{Multi-objective Non-cooperative Game Model}

\subsection{System Model}

A relatively complex computational grid is as Fig. \ref{SM} illustrates. There are $l$ users, $n$ schedulers and $m$ computational nodes:

\begin{itemize}
  \item \textbf{Users}. Users generate tasks to schedulers. Each user generates tasks independently and obeys Poisson distribution.
  \item \textbf{Schedulers}. Schedulers accept tasks from the users, according to the numbers of computational nodes, make a task into slices $a_{ij}(1\geq i\leq n, 1\geq j\leq m)$ ($j$-th slice to computational node $j$ from scheduler $i$). The slice $a_{ij}$ satisfied the constrain of Equation (\ref{ST1}).

      \begin{equation}\label{ST1}
        \sum_{j=1}^m a_{ij}=1(a_{ij}\geq 0)
      \end{equation}
  \item \textbf{Computational Nodes}. Computational nodes really execute task slices. The average processing rate of node $j$ is $u_j$, and the processing time can obey any distribution, so, every computational node can be deemed as an M/G/1 queuing system \cite{24}. Assume that $\lambda_i$ is the average task arriving rate at the scheduler $i$. Two constraints should be satisfied as Equations (\ref{ST2}) and (\ref{ST3}) define.

      \begin{equation}\label{ST2}
        \sum_{i=1}^n\lambda_i < \sum_{j=1}^m u_j
      \end{equation}

      \begin{equation}\label{ST3}
        \sum_{i=1}^n (\lambda_i\cdot a_{ij}) < u_j
      \end{equation}

      The Equation (\ref{ST2}) must be satisfied to prevent tasks to be added infinitely. And the Equation (\ref{ST3}) prevents the task slices allocated to node $j$ to exceed its processing ability.
\end{itemize}

\begin{figure}
    \centering
    \includegraphics{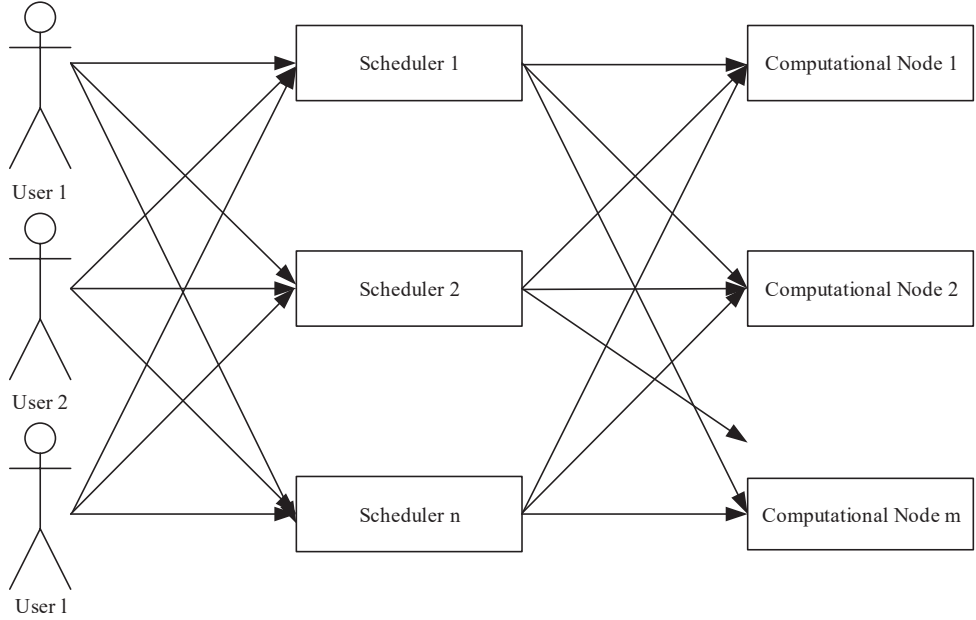}
    \caption{System Model of a Computational Grid}
    \label{SM}
\end{figure}

\subsection{Costs for Tasks in Computational Grid}

Chandy \cite{26} proposed a cost model of tasks in computational grid. It includes four kinds of costs: power cost, network cost, loss cost and utilization cost. Their computational methods is listed in Table \ref{CTCG}. The total costs are the sum of the four kinds of costs, that is, $c_u=u_p+u_n+u_w+u_r$.

\begin{table*}
\centering
\caption{Costs for Tasks in Computational Grid}
\begin{tabular}{|c|c|p{4.5cm}|}
    \hline
    Cost Types  & Computational Methods & Parameter Meanings \\
    \hline
    Power Cost & $u_p=c_p\cdot p\cdot T$ & $c_p$ is cost/Joule, $p$ is the capacity, $T$ is the utilization time of resources.\\
    Network Cost & $u_n=(c_{bw}\cdot b+c_n)\cdot T$ & $c_{bw}$ is the cost of bandwidth, $b$ is the average bandwidth utilized, $c_n$ is the utilization cost within a unit time, $T$ is the utilization time of networks.\\
    Loss Cost & $u_w=\frac{c_r}{MTTF}\cdot T$ & $c_r$ is the cost of the resources, $MTTF$ is the failure time, $T$ is the utilization time of the related resources. \\
    Utilization Cost & $u_r=c_f\cdot\rho\cdot f_r\cdot T$ & $c_f$ is the fixed cost of amortization, $f_r$ is related parts of resources, $\rho$ is the utilization factor of resources, $T$ is the utilization time of resources. \\
    \hline
\end{tabular}

\label{CTCG}
\end{table*}

\subsection{Game Model}

Assume that $B_i$ is the average data length of a task, the transmission time of task slice $a_{ij}$ is as Equation (\ref{TT}) defines, where $e_{ij}$ is the transmission delay between the scheduler $i$ and the computational node $j$, $c_{ij}$ is the bandwidth between the scheduler $i$ and the computational node $j$.

\begin{equation}\label{TT}
L_{ij}=(e_{ij}+\frac{B_i}{c_{ij}})\cdot a_{ij}
\end{equation}

Replacing $T$ by Equation (\ref{TT}) in computational method of network cost $u_n=(c_{bw}b+c_n)T$ in Table \ref{CTCG}, we get Equation (\ref{NC}).

\begin{equation}\label{NC}
u_{nij}=(c_{bw}\cdot c_{ij}+c_n)(e_{ij}+\frac{B_i}{c_{ij}})\cdot a_{ij}
\end{equation}

A computational node in computational grids can be deemed as an M/G/1 queuing system \cite{24}. Task processing time on a computational node includes servicing time and waiting time, their computational methods as Equation (\ref{ST}) and (\ref{WT}) define.

\begin{equation}\label{ST}
F_{ij1}=\overline{h_j}\cdot a_{ij} \cdot \lambda_i
\end{equation}

\begin{equation}\label{WT}
F_{ij2}=\frac{a_{ij}\cdot \lambda_i \cdot\overline{h_j^2}\cdot\sum_{k=1}^n(a_{kj}\cdot\lambda_k)}{2(1-\overline{h_j}\cdot\sum_{k=1}^n(a_{kj}\cdot\lambda_k))}
\end{equation}

Where $\overline{h_j}=\frac{1}{u_j}$ is the average of task servicing time on computational node $j$, $\overline{h_j^2}$ is the variance of task servicing time, and $\lambda_k$ is the average arriving rate of the scheduler $k$.

The power cost of task slice $a_{ij}$ processing on computational node $j$ is composed of the power cost during waiting time and serving time. Replacing $T$ in the power cost computation method $u_p=c_p\cdot p\cdot T$ in Table \ref{CTCG} by Equation (\ref{ST}) and Equation (\ref{WT}), we get the power cost of task slice $a_{ij}$ on node $j$ as Equation (\ref{PC}) defines.

\begin{equation}\label{PC}
u_{pij}=c_p\cdot p_{j1}\cdot \overline{h_j}\cdot a_{ij} \cdot \lambda_i +c_p\cdot p_{j2}\cdot \frac{a_{ij}\cdot \lambda_i \cdot \overline{h_j^2}\cdot \sum_{k=1}^n(a_{kj}\cdot \lambda_k)}{2(1-\overline{h_j}\cdot \sum_{k=1}^n(a_{kj}\cdot \lambda_k))}
\end{equation}

Where $p_{j1}$ is the capacity of node $j$ during servicing time, while $p_{j2}$ is the capacity of node $j$ during waiting time.

The loss cost of task slice $a_{ij}$ processing on node $j$ is determined by the processing time on node $j$. Replacing $T$ by Equation (\ref{ST}) and Equation (\ref{WT}) in the loss cost computation method $u_w=\frac{c_r}{MTTF}\cdot T$ in Table \ref{CTCG}, we get loss cost of task slice $a_{ij}$ on node $j$ as Equation (\ref{LC}) defines.

\begin{equation}\label{LC}
u_{wij}=\frac{c_{rj}}{MTTF_j}\cdot(\overline{h_j}\cdot a_{ij}\cdot \lambda_i+\frac{a_{ij}\cdot\lambda_i\cdot \overline{h_j^2}\cdot\sum_{k=1}^n(a_{kj}\cdot \lambda_k)}{2(1-\overline{h_j}\cdot \sum_{k=1}^n(a_{kj}\cdot\lambda_k))})
\end{equation}

The resource utilization cost of task slice $a_{ij}$ on node $j$ is composed by two parts: one is the CPU utilization cost, and the other is the hard disk utilization cost. The computing percentage of CPU is $f_r=\frac{C_i}{C_j}\cdot a_{ij}$, where $C_i$ is average computation of tasks and $C_j$ is the computation provided by node $j$. The CPU utilization cost occurs during servicing time. The disk utilizing percentage is $f_r=\frac{B_i}{D_j}\cdot a_{ij}$, where $B_i$ is the average bits of tasks and $D_j$ is the disk space of node $j$. The hard disk utilization cost occurs during servicing time and waiting time. So, we get the formula of utilization cost as Equation (\ref{UC}) defines.

\begin{equation}\label{UC}
u_{rij}=c_{fj}\cdot \rho_j\cdot \frac{C_i}{C_j}\cdot a_{ij}\cdot \overline{h_j}\cdot a_{ij}\cdot\lambda_i+c_{fj}\cdot \rho_j\cdot \frac{B_i}{D_j}\cdot a_{ij}\cdot (\overline{h_j}\cdot a_{ij}\cdot\lambda_i+\frac{a_{ij}\cdot\lambda_i\cdot \overline{h_j^2}\sum_{k=1}^n(a_{kj}\cdot \lambda_k)}{2(1-\overline{h_j}\cdot \sum_{k=1}^n(a_{kj}\cdot \lambda_k))})
\end{equation}

So, the total costs of task slice $a_{ij}$ is $c_{uij}=u_{pij}+u_{nij}+u_{wij}+u_{rij}$.

Each scheduler shares the set of computational nodes in a computational grid with each other, it is independent and competes with each other to make its task processing costs minimal. Thus, a non-cooperative game exists among schedulers from $1$ to $n$, and every scheduler acts as a player of the game. The task slicing strategy of scheduler $i$ is $a_i=\{a_{i1},\cdots,a_{ij},\cdots,a_{im}\}$. In the game, every scheduler expects that its task processing cost is minimal. That is, the objective function is as Equation (\ref{OF1}) defines.

\begin{equation}\label{OF1}
MIN(C_{ui})=\left\{
\begin{aligned}
MIN(\sum_{j=1}^m c_{uij}) \\
\textrm{s.t.}\\
\textrm{Equation (\ref{ST1})}\\
\textrm{Equation (\ref{ST2})}\\
\textrm{Equation (\ref{ST3})}
\end{aligned}
\right.
\end{equation}

Where $c_{uij}=u_{pij}+u_{nij}+u_{wij}+u_{rij}$.

\begin{proposition}
The game using Equation (\ref{OF1}) as objective function has a unique Nash equilibrium.
\end{proposition}

\begin{proof}
Just because Equation (\ref{OF1}) is continuous, convex, and increasing, that is,

$$\frac{\partial C_{ui}}{\partial a_{ij}}\geq 0$$

$$\frac{\partial^2 C_{ui}}{\partial a_{ij}^2}\geq 0$$
\end{proof}

But, Equation (\ref{OF1}) is very complex and hard to be solved. We make it to be a multi-objective game as Equation (\ref{OF2}) defines.

\begin{equation}\label{OF2}
MIN(C_{ui})=\left\{
\begin{aligned}
MIN(\sum_{j=1}^m u_{pij}) \\
MIN(\sum_{j=1}^m u_{nij}) \\
MIN(\sum_{j=1}^m u_{wij}) \\
MIN(\sum_{j=1}^m u_{rij}) \\
\textrm{s.t.}\\
\textrm{Equation (\ref{ST1})}\\
\textrm{Equation (\ref{ST2})}\\
\textrm{Equation (\ref{ST3})}
\end{aligned}
\right.
\end{equation}

Where $u_{pij}$ is defined in Equation (\ref{PC}), $u_{nij}$ is defined in Equation (\ref{NC}), $u_{wij}$ is defined in Equation (\ref{LC}), $u_{rij}$ is defined in Equation (\ref{UC}).

\begin{proposition}
The game using Equation (\ref{OF2}) as objective function has a unique Nash equilibrium.
\end{proposition}

\begin{proof}
Just because Equation (\ref{OF2}) is continuous, convex, and increasing, that is,

$$\frac{\partial \sum_{j=1}^m u_{pij}}{\partial a_{ij}}\geq 0\quad \frac{\partial^2 \sum_{j=1}^m u_{pij}}{\partial a_{ij}^2}\geq 0$$

$$\frac{\partial \sum_{j=1}^m u_{nij}}{\partial a_{ij}}\geq 0\quad \frac{\partial^2 \sum_{j=1}^m u_{nij}}{\partial a_{ij}^2}\geq 0$$

$$\frac{\partial \sum_{j=1}^m u_{wij}}{\partial a_{ij}}\geq 0\quad \frac{\partial^2 \sum_{j=1}^m u_{wij}}{\partial a_{ij}^2}\geq 0$$

$$\frac{\partial \sum_{j=1}^m u_{rij}}{\partial a_{ij}}\geq 0\quad \frac{\partial^2 \sum_{j=1}^m u_{rij}}{\partial a_{ij}^2}\geq 0$$
\end{proof}

Equation (\ref{OF2}) is a multi-objective non-cooperative game, following the methods on solving multi-objective game in \cite{27}, we solve $MIN(\sum_{j=1}^m u_{pij})$, and observe other objective's results.

\begin{equation}\label{UPI}
u_{pi}=\sum_{j=1}^m(u_{pij})=\sum_{j=1}^m(c_p\cdot p_{j1}\cdot \overline{h_j}\cdot a_{ij}+c_p\cdot p_{j2}\cdot \frac{a_{ij}\cdot \overline{h_j^2}\cdot \sum_{k=1}^n(a_{kj}\cdot \lambda_k)}{2(1-\overline{h_j}\cdot \sum_{k=1}^n(a_{kj}\cdot \lambda_k))})
\end{equation}

We introduce a new variable $u_{ji}$ as Equation (\ref{UJI}) defines, which denotes the computational power of node $j$ that is available to scheduler $i$.

\begin{equation}\label{UJI}
u_{ji}=u_j-\sum_{k=1,k=\neq i}^n(a_{ij}\lambda_k)
\end{equation}

Combining Equation (\ref{UPI}) and Equation (\ref{UJI}), we get the new equation for power cost as Equation (\ref{UPI2}) defines.

\begin{equation}\label{UPI2}
u_{pi}=\sum_{j=1}^m(c_p\cdot p_{j1}\cdot \overline{h_j}\cdot a_{ij}+c_p\cdot p_{j2}\cdot \frac{a_{ij}\cdot \overline{h_j^2}\cdot (a_{ij}\cdot \lambda_i+\frac{1}{\overline{h_j}}-u_{ji})}{2\overline{h_j}(u_{ji}-a_{ij}\cdot\lambda_i)})
\end{equation}

Since all formulae in Equation (\ref{OF2}) are all convex, the first-order Karush-Kuhn-Tucker conditions are necessary and sufficient to solve it. The Lagrangian is given by

\begin{equation}\label{L}
L=\sum_{j=1}^m(c_p\cdot p_{j1}\cdot \overline{h_j}\cdot a_{ij}+c_p\cdot p_{j2}\cdot \frac{a_{ij}\cdot \overline{h_j^2}\cdot (a_{ij}\cdot \lambda_i+\frac{1}{\overline{h_j}}-u_{ji})}{2\overline{h_j}(u_{ji}-a_{ij}\cdot\lambda_i)}- \alpha\cdot a_{ij})+\alpha
\end{equation}

Let

\begin{equation}
\frac{\partial L}{\partial a_{ij}}=0
\end{equation}

\begin{equation}
\frac{\partial L}{\partial \alpha}=0
\end{equation}

We get the following equations.

\begin{equation}
c_p\cdot p_{j1}\cdot \overline{h_j}+c_p\cdot p_{j2}\cdot (\frac{\overline{h_j^2}u_{ji}}{2\overline{h_j}^2(u_{ji}-a_{ij}\lambda_i)^2}- \frac{\overline{h_j^2}}{2\overline{h_j}})=\alpha
\end{equation}

\begin{equation}
a_{ij}=\frac{u_{ji}}{\lambda_i}-\frac{\sqrt{c_p\cdot p_{j2}\cdot \overline{h_j^2}\cdot u_{ji}}}{\lambda_i\sqrt{\overline{h_j}(c_p\cdot p_{j2}\cdot\overline{h_j^2} +2\overline{h_j}(\alpha-c_p\cdot p_{j1}\cdot\overline{h_j}))}}
\end{equation}

By constraint (\ref{ST1}), $\alpha$ is given by the following equation.

\begin{equation}
\sum_{j=1}^m u_{ji}-\lambda_i=\sum_{j=1}^m \frac{\sqrt{c_p\cdot p_{j2}\cdot \overline{h_j^2}\cdot u_{ji}}}{\sqrt{\overline{h_j}(c_p\cdot p_{j2}\cdot\overline{h_j^2} +2\overline{h_j}(\alpha-c_p\cdot p_{j1}\cdot\overline{h_j}))}}
\end{equation}

We also need to satisfy the constraints (\ref{ST2}) and (\ref{ST3}), and $a_{ij}\geq 0$. Let $a_{ij}=0$, we get the following equation.

\begin{equation}
\alpha=c_p\cdot p_{j1}\cdot\overline{h_j}- \frac{c_p\cdot p_{j2}\cdot \overline{h_j^2}}{2\overline{h_j}} + \frac{c_p\cdot p_{j2}\cdot \overline{h_j^2}}{2\overline{h_j}^2\cdot u_{ji}}
\end{equation}

All computational nodes that make $a_{ij}< 0$ must be excluded, we set $a_{ij}=0$ for these nodes.

\subsection{Solving Algorithm}

Based on the game model in the above subsection, we can design the solving algorithm as follows.

Order the computational nodes according to potential power cost such that $u_{p1i}< u_{p2i}<\cdots<u_{pmi}$, $u_{pji}$ is defined by the following equation.

\begin{equation}
u_{pji}=c_p\cdot p_{j1}\cdot\overline{h_j}- \frac{c_p\cdot p_{j2}\cdot \overline{h_j^2}}{2\overline{h_j}} + \frac{c_p\cdot p_{j2}\cdot \overline{h_j^2}}{2\overline{h_j}^2\cdot u_{ji}}
\end{equation}

Then we have the following equation.

\begin{equation}
a_{ij}=\frac{u_{ji}}{\lambda_i}-\frac{\sqrt{c_p\cdot p_{j2}\cdot \overline{h_j^2}\cdot u_{ji}}}{\lambda_i\sqrt{\overline{h_j}(c_p\cdot p_{j2}\cdot\overline{h_j^2} +2\overline{h_j}(\alpha-c_p\cdot p_{j1}\cdot\overline{h_j}))}}, 1\leq j\leq d_i
\end{equation}

$\alpha$ is given by the following equation.

\begin{equation}
\sum_{j=1}^{d_i} u_{ji}-\lambda_i=\sum_{j=1}^{d_i} \frac{\sqrt{c_p\cdot p_{j2}\cdot \overline{h_j^2}\cdot u_{ji}}}{\sqrt{\overline{h_j}(c_p\cdot p_{j2}\cdot\overline{h_j^2} +2\overline{h_j}(\alpha-c_p\cdot p_{j1}\cdot\overline{h_j}))}}
\end{equation}

$d_i:1 \leq d_i\leq m$ is the maximum positive integer that satisfies (\ref{ST4}).

\begin{equation}\label{ST4}
\sum_{j=1}^{d_i} u_{ji}-\lambda_i \leq \sum_{j=1}^{d_i} \frac{\sqrt{c_p\cdot p_{j2}\cdot \overline{h_j^2}\cdot u_{ji}}}{\sqrt{\overline{h_j}(c_p\cdot p_{j2}\cdot\overline{h_j^2} +2\overline{h_j}((c_p\cdot p_{d_i1}\cdot\overline{h_{d_i}}- \frac{c_p\cdot p_{d_i2}\cdot \overline{h_{d_i}^2}}{2\overline{h_{d_i}}} + \frac{c_p\cdot p_{d_i2}\cdot \overline{h_{d_i}^2}}{2\overline{h_{d_i}}^2\cdot u_{d_ii}})-c_p\cdot p_{j1}\cdot\overline{h_j}))}}
\end{equation}

\section{Experiments}

In this section, we analyze the effects of different aspects on the average task power cost of the schedulers given by (13) and compare the result with the average-allocated algorithm. Besides, with the same parameter of game algorithm, we observe its impacts on other costs.

\subsection{Convergence to Equilibrium of the Game Algorithm}

In this experiment, we analyze the convergence to equilibrium of the game algorithm. Each scheduler chooses its own strategy independently and dependents on a particular system state. However, the system state is always changing when schedulers are running, so that the strategy for the scheduler needs to be updated when the system state is changing. In order to make the system reach to a stable state, whereby no player has a tendency to unilaterally change its strategy, it needs to make the algorithm run iteration and reach the Nash equilibrium eventually. The initial strategy $ {a_i} $ of each scheduler $i$ is the zero vector, each scheduler then refines and updates its strategy at each iteration. When the result does not change, we expect the system to reach a Nash equilibrium.

In this experiment, we set the average system load to 0.2. The values in Table \ref{TPROCESSOR} represent the average processing rate of each computational node and the relative job arrival rate for each scheduler is shown in Table \ref{TARS}. The result, as shown in Fig. \ref{CGA} that the algorithm converges to a Nash equilibrium in 4 iterations (in this paper, we assume that convergence has occurred when the overall percentage change is 0.0001). In terms of the periodic scheduling done by each scheduler, an equilibrium is reached when the calculated strategy does not change from one iteration to another.

\begin{table}
\centering
  \caption{Relative processing rate of the computational nodes}
  \begin{tabular}{|c|c|c|c|c|c|c|c|c|}
    \cline{1-9}
    computational node  & 1 & 2 & 3 & 4 & 5 & 6 & 7 & 8 \\
    \cline{1-9}
    $u_j$(the average processing rate of node j) & 35 & 46 & 37 & 28 & 29 & 30 & 41 & 32\\
    \cline{1-9}
  \end{tabular}

  \label{TPROCESSOR}

\end{table}

\begin{table}
\centering
  \caption{Relative task arriving rate at the schedulers}
  \begin{tabular}{|c|c|c|c|c|c|}
    \cline{1-6}
    computational node  & 1 & 2 & 3 & 4 & 5 \\
    \cline{1-6}
    $\lambda_i$(the average task arriving rate at the scheduler i) & 6.672 & 2.78 & 3.336 & 6.672 & 5.004\\
    \cline{1-6}
  \end{tabular}

  \label{TARS}
\end{table}

\begin{figure}
    \centering
    \includegraphics[scale=0.3]{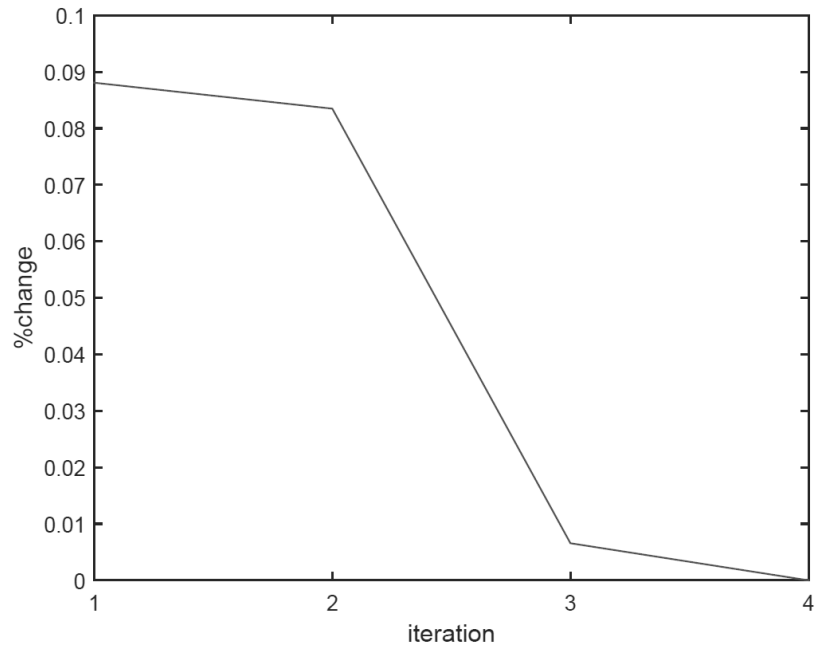}
    \caption{Convergence of the game algorithm}
    \label{CGA}
\end{figure}

The average power cost per task for each scheduler, when the system is at equilibrium, is shown in Fig. \ref{AC}. The power cost includes the cost of execution time of the task itself, the cost of waiting time at the queue given by (13). As shown in the figure, the average power cost per task for each scheduler, for both the game algorithm and average schemes, is normalized by dividing the sum of average task cost of all schedulers in the game algorithm. As can be seen from the figure, game algorithm has a lower cost than average scheme for every scheduler.

\begin{figure}
    \centering
    \includegraphics[scale=0.5]{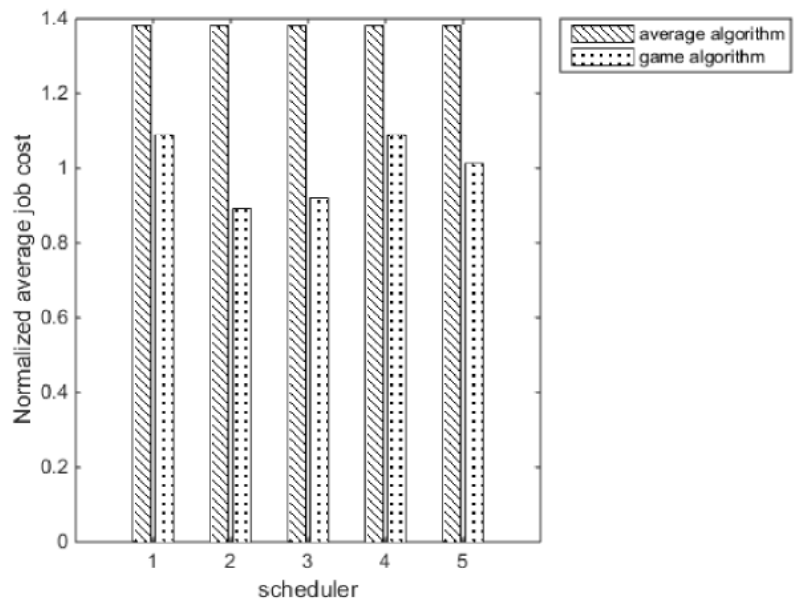}
    \caption{Average task cost for each scheduler}
    \label{AC}
\end{figure}

\subsection{Effect of System Loads}

In this set of experiments, we vary the average system load from 0.1 to 0.9. The same set of computational nodes and schedulers are used as in the previous set of experiments. The arrival rate of tasks for each scheduler is then adjusted to give the required average system load.

Fig. \ref{ACSL} shows the normalized average job costs as the system load is varied from 0.1 to 0.9. As before, the job cost is normalized by dividing each cost by the overall average cost of the game scheme. In this figure, we see an increase in the system wide average job cost as the system loads increase. This trend is explained by and is a consequence of (15). As the system loads increase, the average queue length at the computational node gets longer, and as a result, the average cost will be added by cost when the task is waiting in the queue. Both of the game algorithm and average scheme show the same trend, although the game algorithm gives lower expected costs.

\begin{figure}
    \centering
    \includegraphics[scale=0.3]{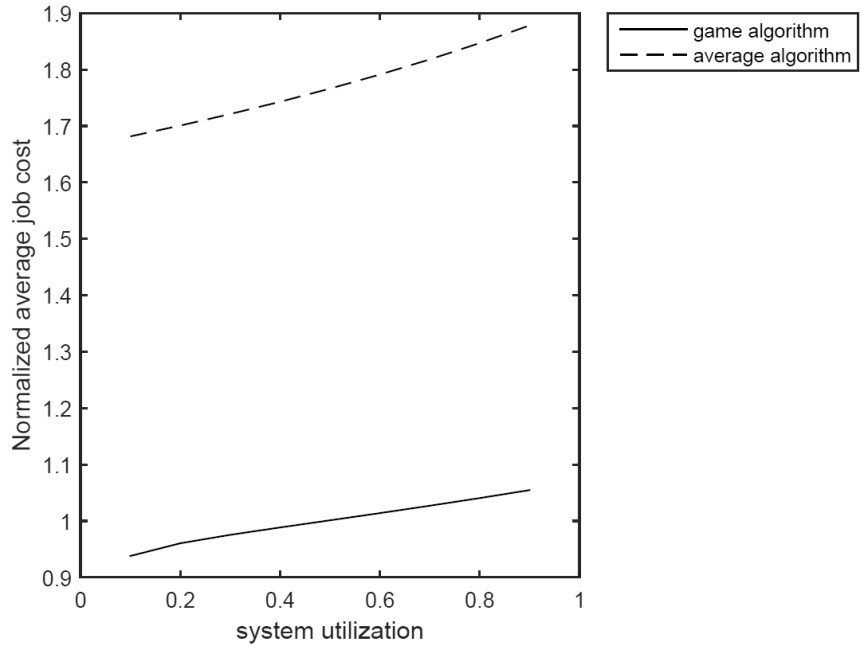}
    \caption{Average task cost versus system load}
    \label{ACSL}
\end{figure}

\subsection{Effect of System Size}

In this part of the experiment, we vary the number of computational nodes and the number of schedulers in the system and investigate its effect on the average cost of the schedulers for both the game algorithm and average schemes.

First, we vary the number of computational nodes in the system from 5 to 16. The processing rate of each of the computational nodes is shown in Table \ref{TPROCESSOR2}. We keep the average system load as 20 percent. The result of the effect of system size with computational nodes on the average cost of the schedulers is shown in Fig. \ref{ACSS}. As before, the job cost is normalized by dividing each cost by the average cost of the game algorithm. As can be seen in the figure, for both the game algorithm and average schemes, the average cost decreases as the number of computational nodes in the system increases. Fig. \ref{ACSS} also demonstrates that the game algorithm results in a lower overall average cost than the average scheme over system size ranging from 5 to 16 computational nodes. This shows that an efficient allocation of tasks to the computational nodes is important in grid systems having multiple computational nodes.

\begin{center}
\begin{table}
  \begin{tabular}{|c|c|c|c|c|c|c|c|c|c|c|c|c|c|c|c|c|}
    \cline{1-17}
    computational node  & 1 & 2 & 3 & 4 & 5 & 6 & 7 & 8 & 9 & 10 & 11 & 12 & 13 & 14 & 15 & 16 \\
    \cline{1-17}
    $u_j$(the average processing rate of node j) & 35 & 46 & 37 & 28 & 29 & 30 & 41 & 32 & 35 & 46 & 40 & 39 & 41 & 30 & 41 & 32\\
    \cline{1-17}
  \end{tabular}
  \caption{Relative processing rate of the computational nodes}
  \label{TPROCESSOR2}
\end{table}
\end{center}

\begin{figure}
    \centering
    \includegraphics[scale=0.3]{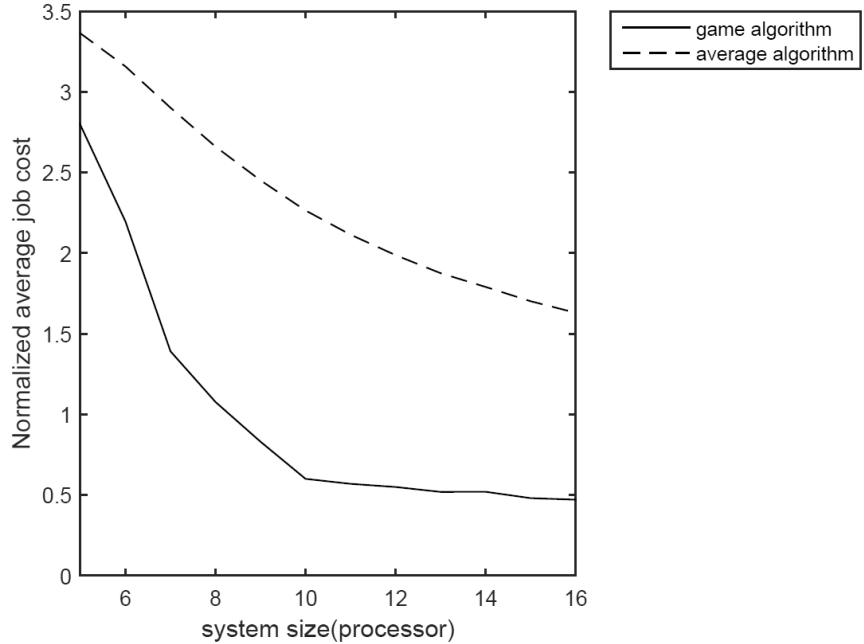}
    \caption{Average task¡¯s cost versus system size(computational node)}
    \label{ACSS}
\end{figure}

Then, we vary the number of schedulers in the system from 2 to 10. We keep the total number of tasks arrived to the schedulers in the system the same as the number of schedulers increasing. As such, we can analyze the effects on the average cost of the schedulers for both the game algorithm and average schemes. The result of the effect of system size with schedulers on the average cost of the schedulers is shown in Fig. \ref{ACSS2}. As can be seen in the figure, for the average scheme, the average cost keeps unchanged as the number of schedulers in the system increases. This is because total tasks are kept the same so the numbers of tasks, which are sent to each computational node kept the same, then the average cost would not change. But for the game algorithm, the figure shows a different result, which average cost of schedulers are different, as the number of schedulers are different. When there are 7 schedulers in the system, the result can be best that the average cost of schedulers can be lowest.

\begin{figure}
    \centering
    \includegraphics[scale=0.3]{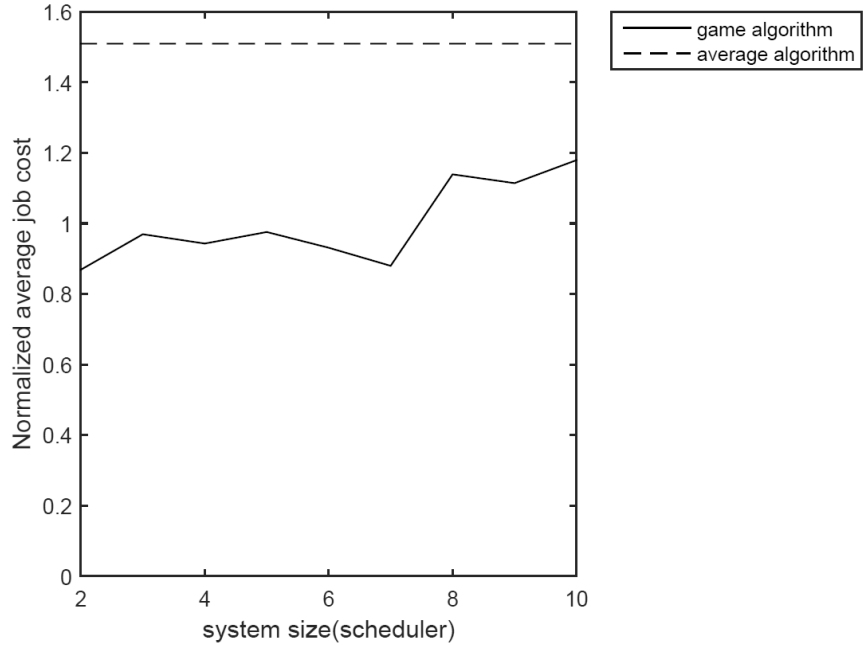}
    \caption{Average task cost versus system size(scheduler)}
    \label{ACSS2}
\end{figure}

\subsection{Effect of Service Time}

In the first set of experiments, we assume that the task service times follow an exponential distribution. However, it has been suggested that the service time of tasks for certain applications follows a heavy-tailed distribution, instead of an exponential distribution. One of the most common distributions used to model such a heavy-tailed distribution is the Bounded Pareto distribution.

The Bounded Pareto distribution is characterized by the following probability density function (pdf):

\begin{equation}
f(x)=\frac{\alpha k^\alpha \cdot x^{-\alpha -1}}{1-(\frac{k}{p})^\alpha}
\end{equation}

where $k$ is the minimum job execution time and $p$ is the maximum job execution time; the parameter ¦Á defines the shape of the hyperbolic curve of the distribution.
The mean (first moment) of the distribution is given by

\begin{equation}
\overline{h}=\frac{\alpha}{\alpha -1} \frac {k^\alpha}{(\frac {k}{p})^\alpha}(\frac {1}{k^{\alpha -1}}-\frac {1}{p^{\alpha -1}})
\end{equation}

and the second moment is given by

\begin{equation}
\overline{h^2}=\frac{\alpha}{\alpha -2} \frac {k^\alpha}{(\frac {k}{p})^\alpha}(\frac {1}{k^{\alpha -2}}-\frac {1}{p^{\alpha -2}})
\end{equation}

As before, we use 8 computational nodes in this set of experiments. The parameters $k$ and $p$ used for each of the computational nodes are shown in Table \ref{KVEP} The hyperbolic curve parameter of the Bounded Pareto distribution is then set to $\alpha$= 1.1. Table \ref{RPPTP} summarizes the above values in terms of the expected task execution time at each of the computational nodes; Table \ref{RPPTP} also shows the variance of the task execution time at each of the computational nodes.

The same set of schedulers and parameters are then used as in the previous set of experiments. We then vary the system loads from 0.1 to 0.9 and investigate the effect of the Bounded Pareto service time of tasks on the number of iterations required to reach equilibrium for the game algorithm. The effect of the Bounded Pareto service times on the overall average task power cost are shown in Fig. \ref{STSL}. As can be seen in the figure, the trend in the system wide average job power cost is similar with the previous result shown in Fig. \ref{ACSL}, where the service times follow an exponential distribution. Besides, the game algorithm gives a lower overall task power cost than the average-allocated algorithm.

\begin{center}
\begin{table}
  \begin{tabular}{|c|c|c|c|c|c|c|c|c|}
    \cline{1-9}
    computational node  & 1 & 2 & 3 & 4 & 5 & 6 & 7 & 8 \\
    \cline{1-9}
    k & 0.001 & 0.001 & 0.002 & 0.002 & 0.003 & 0.003 & 0.004 & 0.004\\
    \cline{1-9}
    p & 0.07 & 0.07 & 0.08 & 0.08 & 0.09 & 0.09 & 0.1 & 0.1 \\
    \cline{1-9}
  \end{tabular}
  \caption{k value for each computational node}
  \label{KVEP}
\end{table}
\end{center}

\begin{center}
\begin{table}
  \begin{tabular}{|c|c|c|c|c|c|c|c|c|}
    \cline{1-9}
    computational node  & 1 & 2 & 3 & 4 & 5 & 6 & 7 & 8 \\
    \cline{1-9}
    meas(s) & 0.003843 & 0.003843 & 0.006906 & 0.006906 & 0.009746 & 0.009746 & 0.012471 & 0.012471\\
    \cline{1-9}
    variance($s^2$) & 5.52E-05 & 5.52E-05 & 0.000133 & 0.000133 & 0.000229 & 0.000229 & 0.000345 & 0.000345 \\
    \cline{1-9}
  \end{tabular}
  \caption{relative Processing Power of the Computational Nodes}
  \label{RPPTP}
\end{table}
\end{center}

\begin{figure}
    \centering
    \includegraphics[scale=0.3]{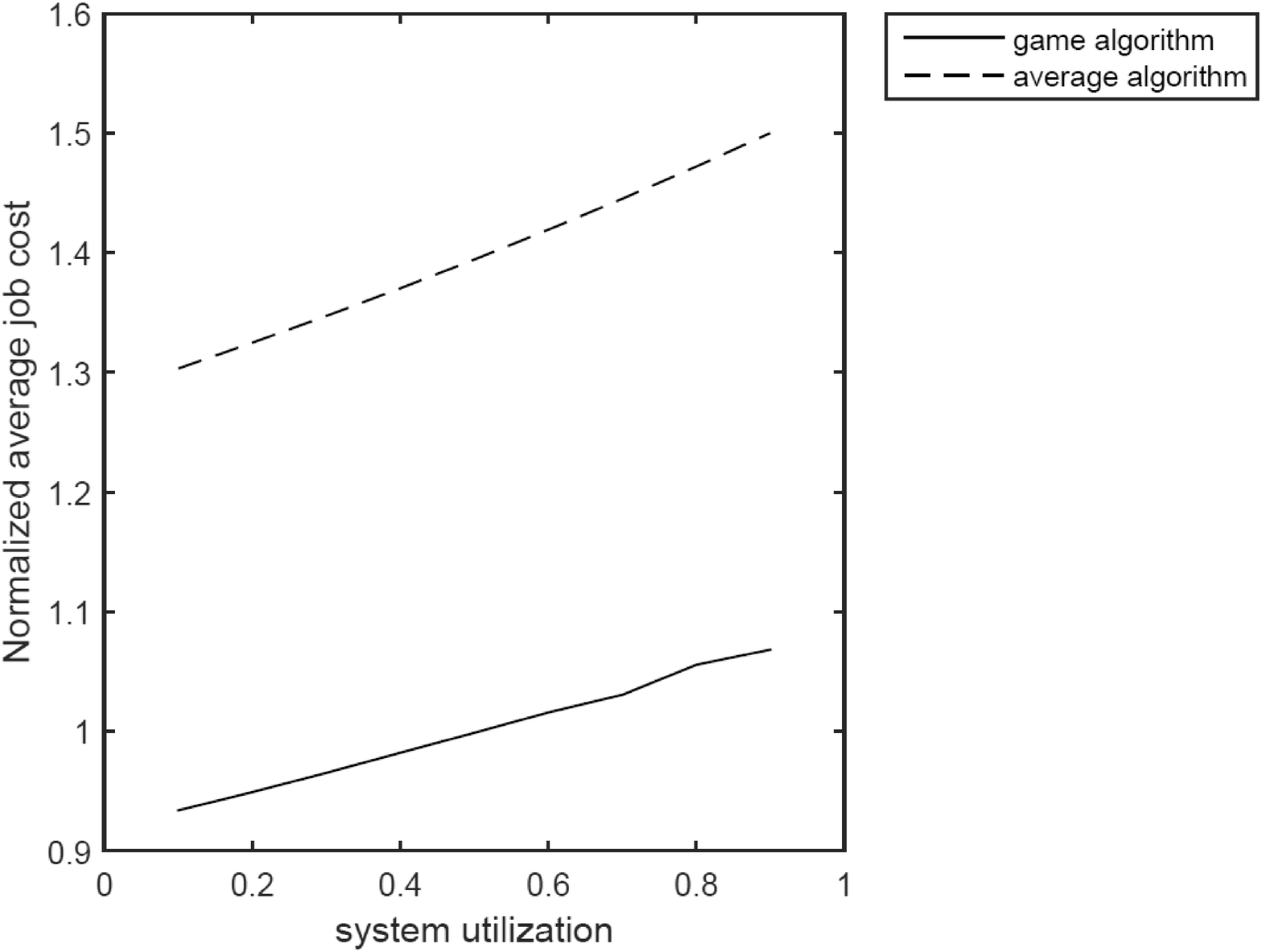}
    \caption{Average task¡¯s power cost versus system load}
    \label{STSL}
\end{figure}
\subsection{Fairness}

In this part of the experiment, we investigate the fairness of each of the different schemes. Fairness would be achieved when the average task power cost for each of the scheduler is the same. If one scheduler has a lower average task power cost and another has a higher average task power cost, then the scheduling scheme can be considered unfair, as it gives some schedulers an advantage and other schedulers a disadvantage.

A fairness index given by

\begin{equation}
FI=\frac{(\sum_{i=1}^n T_i)^2}{n \sum_{i=1}^n {T_i}^2}
\end{equation}

where $T_i$ is the average task power cost of scheduler $i$. If a load-balancing scheme is 100 percent fair, then $FI$ is 1.0. A fairness index close to 1.0 indicates a relatively fair load-balancing scheme

In the first part of the experiment, we vary the average system load from 0.1 to 0.9. The results are shown in Fig. \ref{FVSL1}. As can be seen in the figure, the average-allocated algorithm has a fairness index of 1.0 across the entire utilization range from 0.1 to 0.9. This is the inherent advantage of the average-allocated algorithm¡ªeven though it is a distributed, decentralized scheme and has more cost compared with the game algorithm, it guarantees the same average task power cost for each of the schedulers. As shown in the figure, the game algorithm decreases in ¡°fairness¡± as the system nears full capacity. However, the fairness index at 90 percent system load is still above 0.98, and depending on the requirement of the application, this value may be more optimal than the minimum acceptable level.

In the next set of experiments, we set the average system load to 20 percent and vary the number of computational nodes in the system from 2 to 8. The results are shown in Fig. \ref{FVSS1}. As in the previous experiment, the average-allocated algorithm gives a fairness index of 1.0 as the number of computational nodes is varied from 2 to 8. The game algorithm shows some variations as the number of computational nodes is varied. As before though, the fairness index in all of the cases is above 0.99, which, depending on the application, may be better than the minimum acceptable level.

In this part of the experiment, we change the set of schedulers from a highly heterogeneous set of schedulers shown in Table 2 to a less heterogeneous set of schedulers, as shown in Table 6. The results are shown in Figs. \ref{FVSL2} and \ref{FVSS2}. As can be seen in the figures, using a less heterogeneous set of schedulers has improved the fairness of the game algorithm as compared to the previous set of experiments.

\begin{figure}
    \centering
    \includegraphics[scale=0.3]{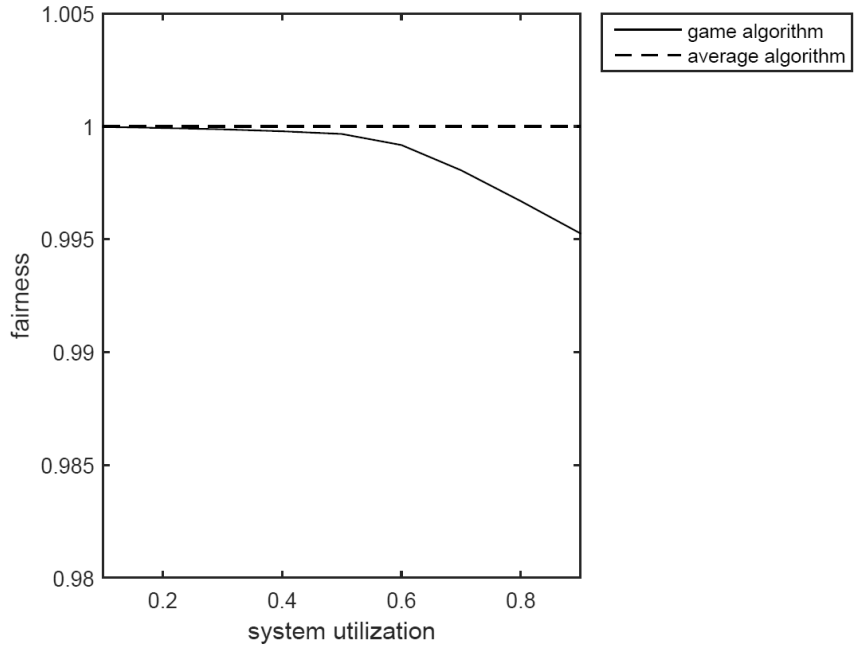}
    \caption{Fairness versus system load}
    \label{FVSL1}
\end{figure}

\begin{figure}
    \centering
    \includegraphics[scale=0.3]{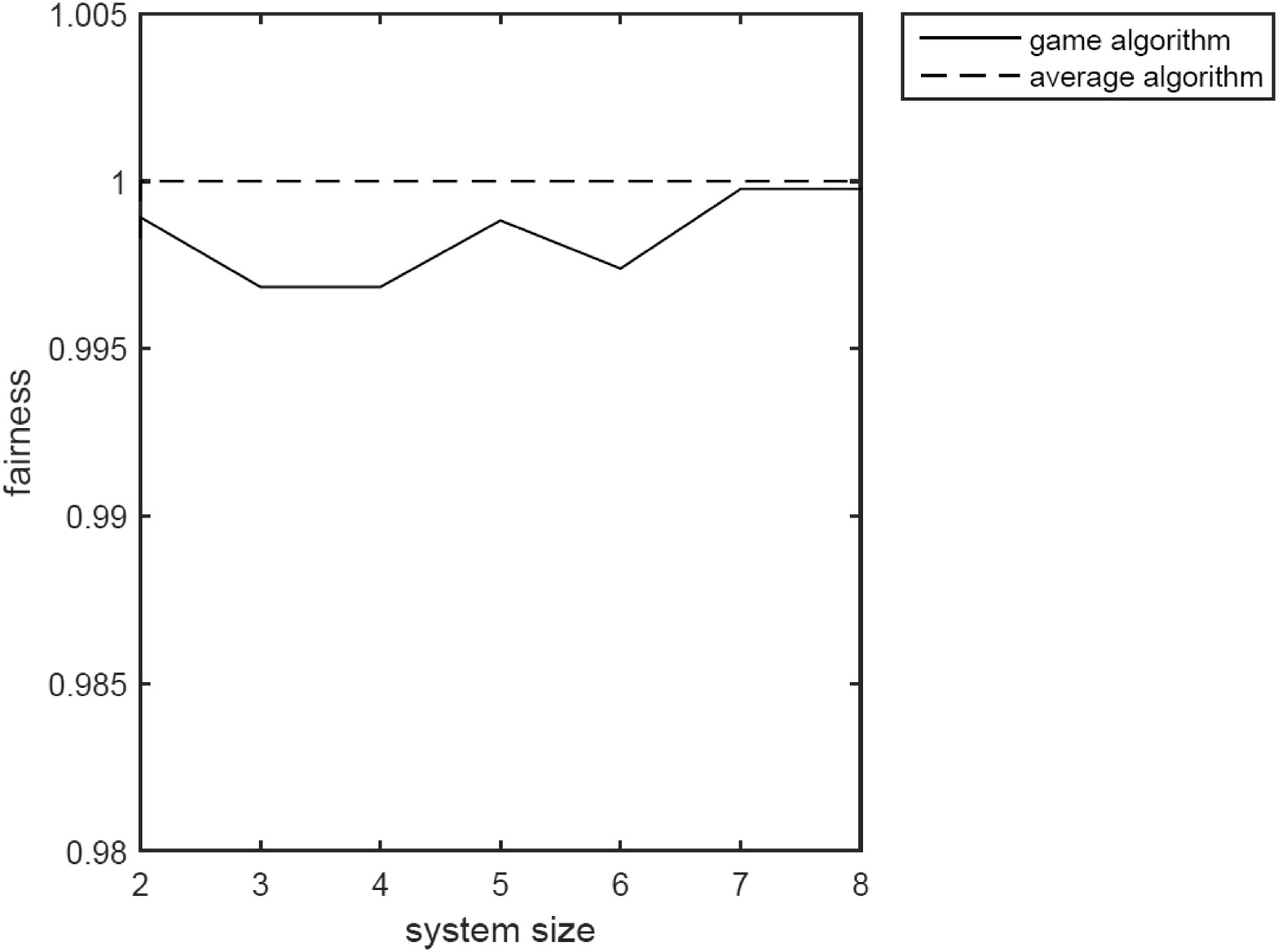}
    \caption{Fairness versus system size}
    \label{FVSS1}
\end{figure}

\begin{figure}
    \centering
    \includegraphics[scale=0.3]{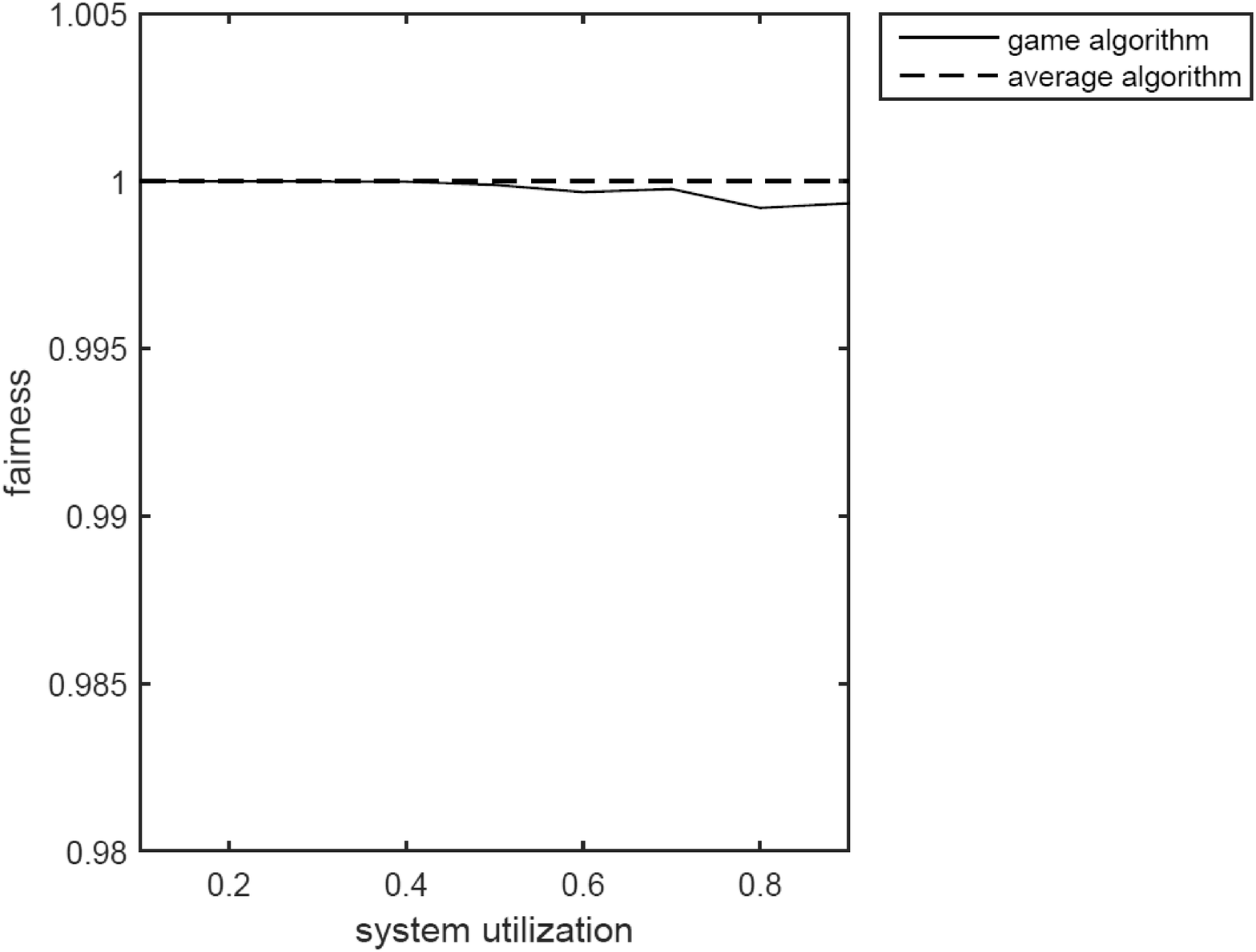}
    \caption{Fairness versus system load}
    \label{FVSL2}
\end{figure}

\begin{figure}
    \centering
    \includegraphics[scale=0.3]{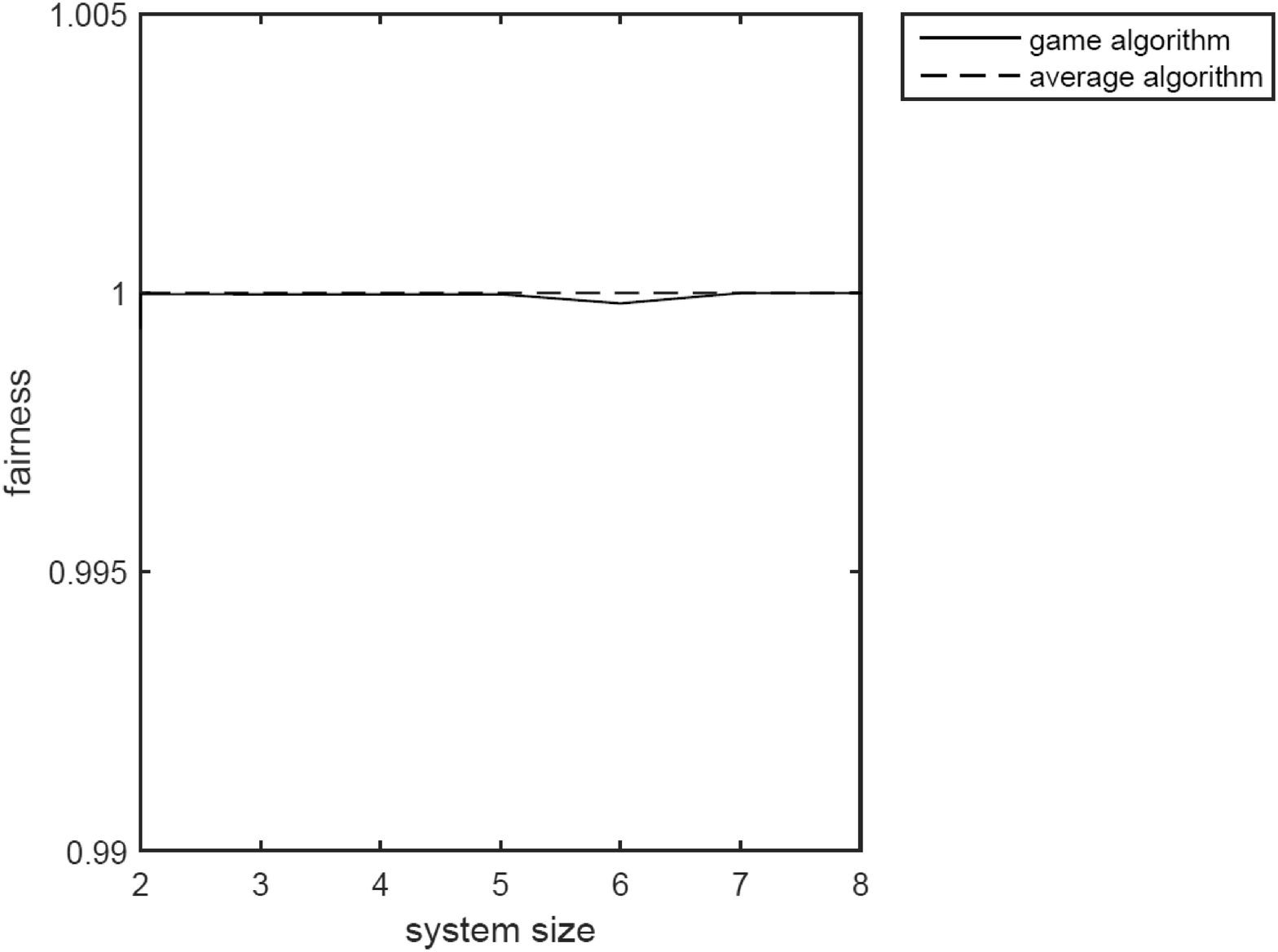}
    \caption{Fairness versus system size}
    \label{FVSS2}
\end{figure}

\subsection{Impact on other costs}

From these experiments above, we can find that the game algorithm has a strong advantage in power cost of scheduler. Followed by the analysis of power cost above with the same set of parameters and the same objective equation, we observe its impact on other objectives, that is network cost, loss cost, and utilization cost.

These results  are shown in Fig. \ref{ANC}, Fig.\ref{ALC}, Fig. \ref{AUC} for network cost, lost cost, and utilization cost respectively. As before, the job cost is normalized by dividing each cost by the average cost of the game scheme. As can be seen from figures, game algorithm has a lower cost than average scheme for every scheduler no matter on network cost or lost cost, or utilization cost. Which means that the result of multi-objective non-cooperative game, what we define in equation 12, is reasonable.

\begin{figure}
    \centering
    \includegraphics[scale=0.5]{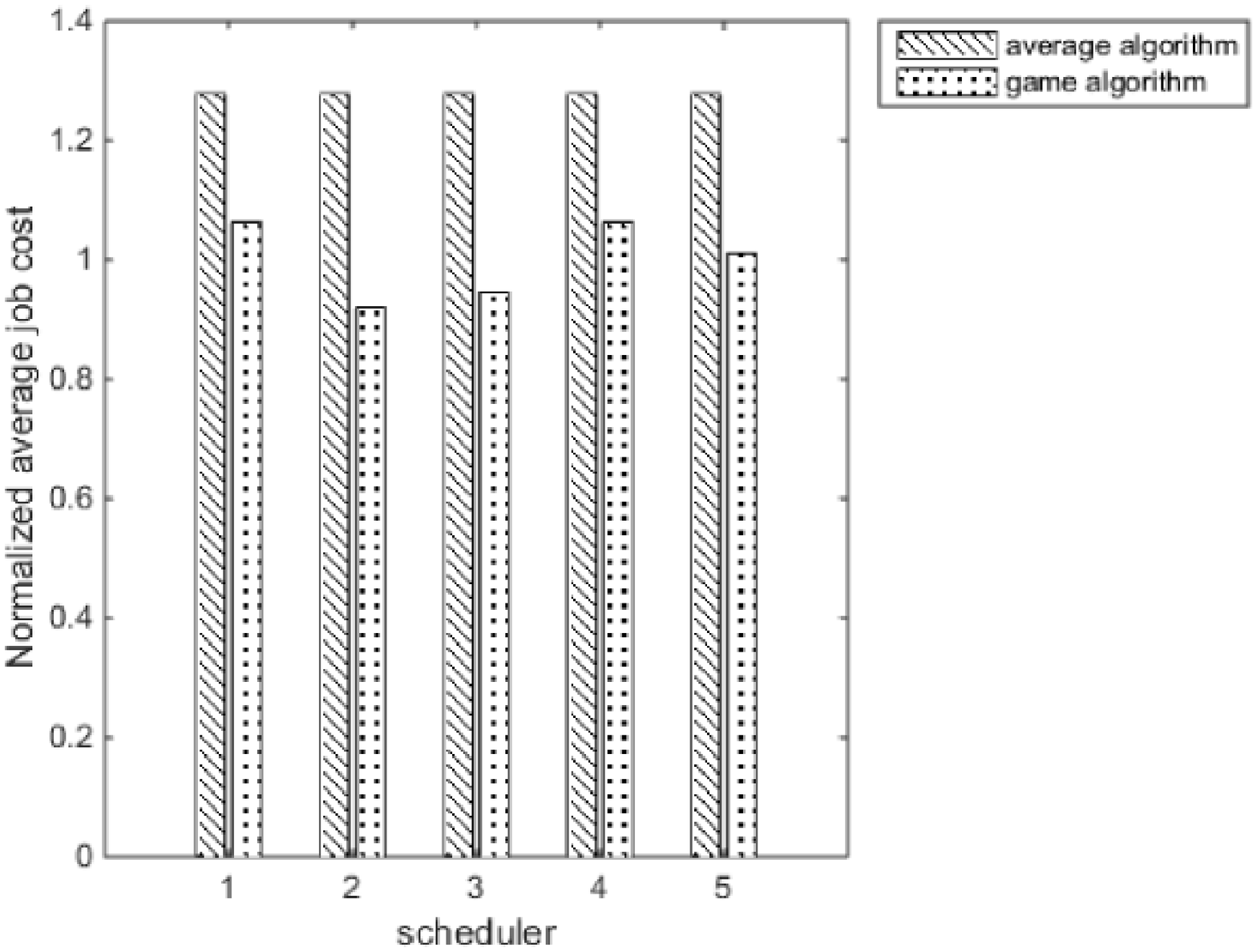}
    \caption{Average task network cost for each scheduler}
    \label{ANC}
\end{figure}

\begin{figure}
    \centering
    \includegraphics[scale=0.5]{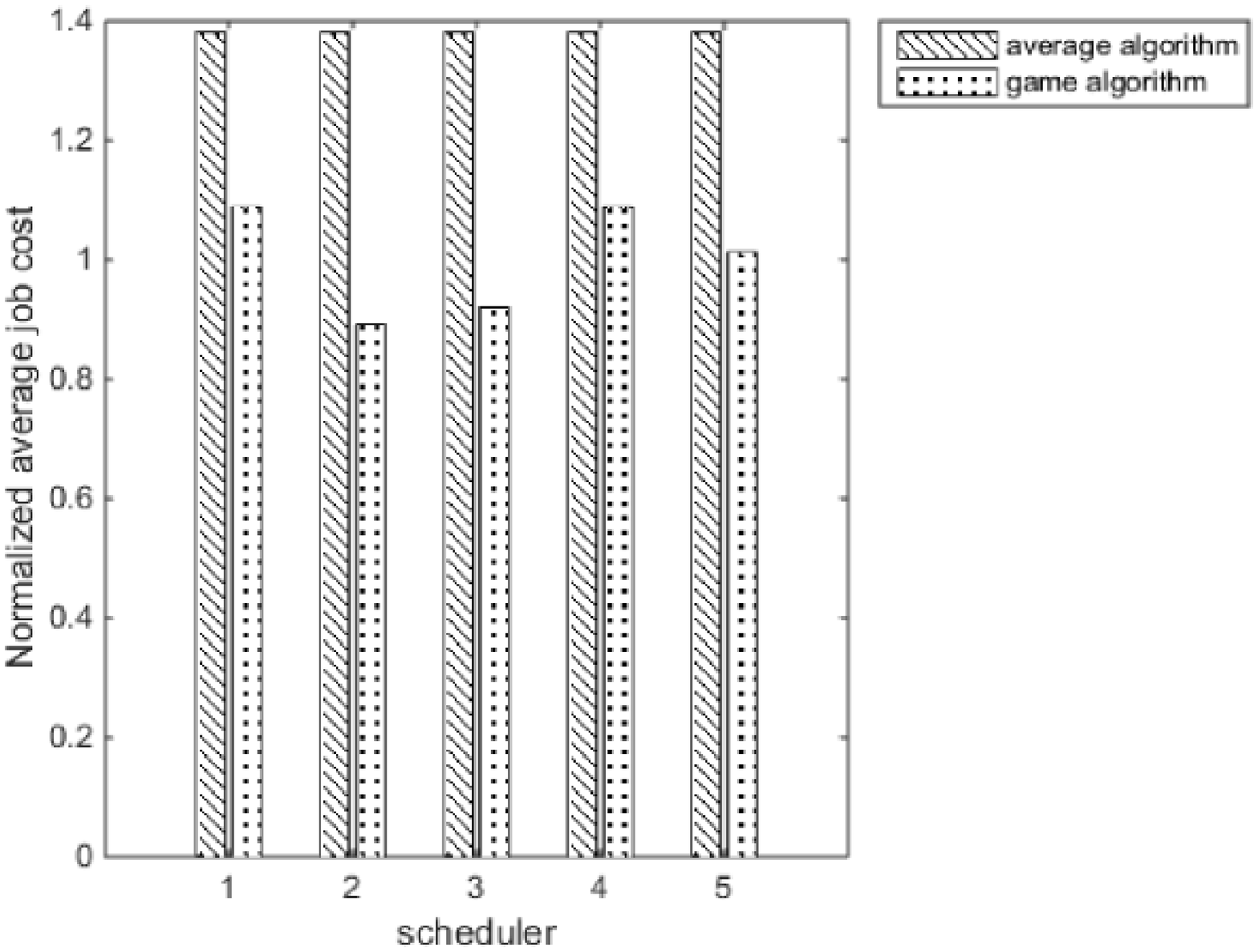}
    \caption{ Average task lost cost for each scheduler}
    \label{ALC}
\end{figure}

\begin{figure}
    \centering
    \includegraphics[scale=0.5]{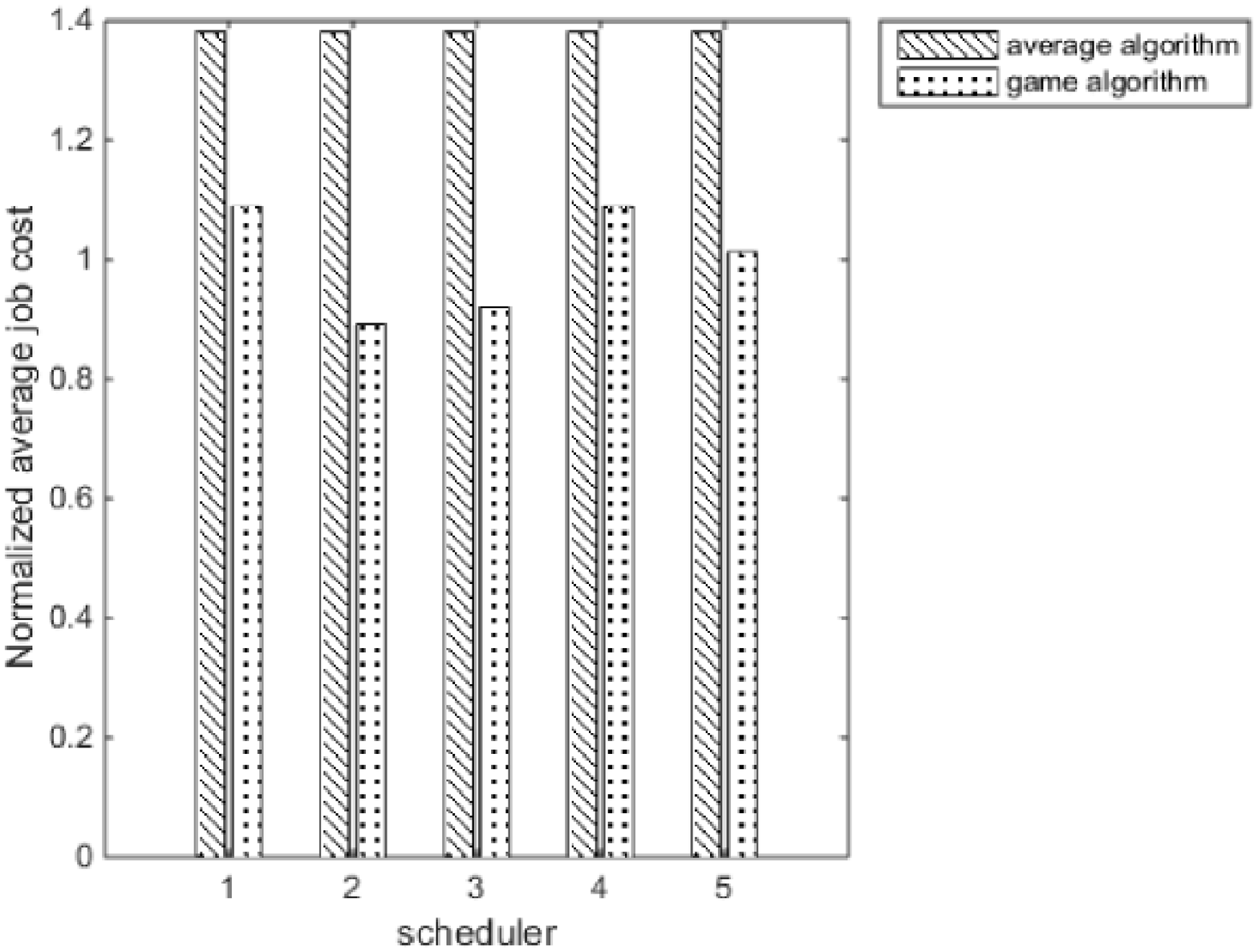}
    \caption{Average task utilization cost for each scheduler}
    \label{AUC}
\end{figure}

\section{Conclusions}

In this paper, we propose a game theoretic algorithm that solves the grid load balancing problem. It aims at minimizing the average task cost for schedulers when tasks are executed in the grid system. The algorithm is semi-static and responds to the changes in system states during runtime. This game algorithm does not assume any particular distribution for service times. It can run correctly only with the first moment and second moment of service times. The experiment results show that the algorithm has a lower cost in the schedule of computational grid.

%
%

\label{lastpage}

\end{document}